\documentclass[12pt]{article}
\usepackage[a4paper,left=1in,right=1in,top=1in,bottom=1in,headheight=28pt]{geometry}
\usepackage{mdframed}
\usepackage{graphicx}
\usepackage{cite}
\usepackage{subfig}
\usepackage{amssymb}
\usepackage{amsmath}
\usepackage{amsthm}
\usepackage{amsfonts}
\usepackage{booktabs}
\usepackage{mathrsfs}
\usepackage{mathtools}
\usepackage{enumitem}
\usepackage{hyperref}
\usepackage{algorithmic}
\usepackage{microtype}
\usepackage{tabularray}
\usepackage[capitalize]{cleveref}

\newtheorem{theorem}{Theorem}[section]
\newtheorem{lemma}[theorem]{Lemma}

\newtheorem{corollary}[theorem]{Corollary}
\newtheorem{observation}[theorem]{Observation}
\theoremstyle{definition}
\newtheorem{definition}[theorem]{Definition}

\DeclareMathOperator{\lcm}{lcm}

\hypersetup{
    colorlinks=true,
    linkcolor=black,
    citecolor=black,
    filecolor=black,
    urlcolor=black,
    pdftitle={Distributed Half-Integral Matching and Beyond},
    pdfauthor={Sameep Dahal and Jukka Suomela}
}

\begin{document}
\title{Distributed Half-Integral Matching and Beyond}

\author{
    Sameep Dahal and Jukka Suomela \\
    Aalto University
}

\date{}

\maketitle

\begin{abstract}
    By prior work, it is known that any distributed graph algorithm that finds a \emph{maximal matching} requires $\Omega(\log^* n)$ communication rounds, while it is possible to find a \emph{maximal fractional matching} in $O(1)$ rounds in bounded-degree graphs.
    However, all prior $O(1)$-round algorithms for maximal fractional matching use arbitrarily fine-grained fractional values. In particular, none of them is able to find a \emph{half-integral} solution, using only values from $\{0, \frac12, 1\}$.
    We show that the use of fine-grained fractional values is necessary, and moreover we give a \emph{complete characterization} on exactly how small values are needed: if we consider maximal fractional matching in graphs of maximum degree $\Delta = 2d$, and any distributed graph algorithm with round complexity $T(\Delta)$ that only depends on $\Delta$ and is independent of $n$, we show that the algorithm has to use fractional values with a denominator at least $2^d$. We give a new algorithm that shows that this is also sufficient.
\end{abstract}

\section{Introduction}

By prior work, it is known that there is a distributed graph algorithm that finds a \emph{maximal fractional matching} (see \cref{ssec:intro-fract}) in $O(\Delta)$ rounds in graphs of maximum degree $\Delta$ \cite{VC-SC}; in particular, the running time is independent of $n$ and only depends on $\Delta$. However, the algorithm uses very fine-grained fractional values; when $\Delta$ increases, the denominators grow exponentially fast. In this work we show that this is necessary: any distributed graph algorithm that finds a maximal fractional matching in $T(\Delta)$ rounds, independently of $n$, has to use fractional values with a denominator at least $2^{\lfloor \Delta/2 \rfloor}$ (and this is tight). In particular, there cannot be a $T(\Delta)$-rounds algorithm for finding a maximal \emph{half-integral} matching.

\subsection{Distributed maximal matching is hard}\label{ssec:intro-int}

Maximal matching is one of the classic problems in the field of distributed graph algorithms, studied extensively since the very early days of the field in the 1980s \cite{%
    ISRAELI-ITAI,%
    MM3,%
    Improved-Deterministic-Distributed-Matching-via-Rounding,%
    Conference-Locality-of-Symmetry-Breaking,%
    ACM-Journal-Locality-of-Symmetry-Breaking,%
    MM2,
    MM4%
}. In the maximal matching problem, the task is to find a matching (a set of edges without common vertices) that is not a strict subset of another matching. This is something one can trivially find in a centralized setting (pick independent edges greedily until you are stuck), but this is a challenging coordination task in a distributed setting, for two reasons:
\begin{enumerate}
    \item One has to \emph{break symmetry}. For example, if the input graph is a cycle, one has to select some but not all edges---the input is symmetric, but the output is not. The task is not solvable at all without resorting to, e.g., unique identifiers or randomness, and even then we cannot solve the task in constant number of rounds; maximal matching in cycles requires $\Omega(\log^* n)$ rounds \cite{Linial-Seminal-Paper,Naor1991}.
    \item One has to solve a \emph{local coordination} task. Even if we have a $\Delta$-regular bipartite graph, with the bipartition given, we still need $\Omega(\Delta)$ rounds to find a maximal matching, at least in sufficiently large graphs \cite{MM4}.
\end{enumerate}
On the positive side, $O(\Delta + \log^* n)$-round distributed algorithms for finding a maximal matching in a graph of maximum degree $\Delta$ are known \cite{MM3}; one can also make different trade-offs between dependency on $\Delta$ vs.\ $n$ \cite{%
    Improved-Deterministic-Distributed-Matching-via-Rounding,%
    Conference-Locality-of-Symmetry-Breaking,%
    ACM-Journal-Locality-of-Symmetry-Breaking%
}, but it is impossible to achieve a running time of $T(\Delta)$, independent of $n$ \cite{Linial-Seminal-Paper,Naor1991}. All of these results hold in the usual LOCAL model of distributed computing (see \cref{ssec:prelim-model} for the details).

\subsection{Distributed fractional matching is easier}\label{ssec:intro-fract}

A matching $M \subseteq E$ in a graph $G = (V,E)$ can be interpreted as a function $x$ that assigns value $x(e) = 1$ to each edge $e \in M$. If we let \[x[v] = \sum_{e \in E: v \in e} x(e)\] denote the sum of labels on edges incident to node $v \in V$, then we can define that function $x\colon E \to \{0,1\}$ is a \emph{matching} if $x[v] \le 1$ for all $v \in V$. Moreover, $x$ is a \emph{maximal matching} if for each edge $\{u,v\} \in E$ at least one endpoint is \emph{saturated}, i.e., $x[u] = 1$ or $x[v] = 1$. Finally, $x$ is a \emph{maximum matching} if it maximizes $\sum_e x(e)$.

\begin{figure}[b]
    \centering
    \includegraphics[page=1]{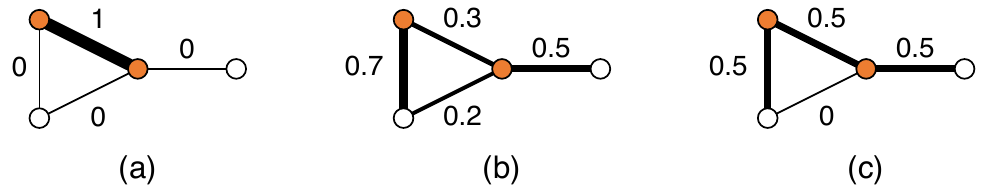}
    \caption{(a) A maximal matching. (b) A maximal fractional matching. (c) A maximal half-integral matching. The orange nodes are saturated.}\label{fig:mm-mfm}
\end{figure}

We can now also consider the \emph{fractional relaxation} of this integer program. We say that $x\colon E \to [0,1]$ is a \emph{fractional matching} if it satisfies $x[v] \le 1$ for each $v \in V$, it is a \emph{maximal fractional matching} if $x[u] = 1$ or $x[v] = 1$ for each edge $\{u,v\} \in E$, and it is a \emph{maximum fractional matching} if it maximizes $\sum_e x(e)$. See \cref{fig:mm-mfm} for illustrations.

Note that any maximal matching is also a maximal fractional matching, but the converse is not necessarily true. However, maximal fractional matchings share many useful properties of maximal matchings. For example, the set of saturated nodes forms a $2$-approximation of a minimum vertex cover \cite{VC-Approx-Linear-Time}.

When we consider distributed graph algorithms for maximal fractional matchings, one of the obstacles discussed in \cref{ssec:intro-int} goes away: \emph{we do not need to break symmetry}. For example, if the graph is a cycle, we can simply label all edges with $1/2$. More generally, if we have a $d$-regular graph, we can label all edges with $1/d$. The lower bound of $\Omega(\log^* n)$ from \cite{Linial-Seminal-Paper,Naor1991} for symmetry-breaking problems no longer applies.

While the case of non-regular graphs is much more challenging, it is nevertheless possible to design distributed algorithms that find a maximal fractional matching in $O(\Delta)$ rounds, independently of $n$ \cite{VC-SC}. It is also known that the local coordination challenge does not disappear; $o(\Delta)$-round algorithms do not exist \cite{Linear-in-Delta}.

\subsection{What about half-integral matchings?}\label{ssec:intro-half}

The fractional matching polytope is \emph{half-integral} (see e.g.\ \cite[Section 30.3]{schrijver-book}). That is, there exists a maximum fractional matching in which $x(e) \in \{0, \frac12, 1\}$ for every edge $e \in E$.

There is also a simple distributed strategy that at first seems to lead to half-integral solutions (see e.g.\ \cite{Delta-squared}). First, construct the \emph{bipartite double cover} $G' = (V',E')$ of the graph $G = (V,E)$: for each node $v \in V$ we have two nodes $v_1$ and $v_2$ in $V'$, and for each edge $\{u,v\} \in E$ we have two edges $\{u_1, v_2\}$ and $\{u_2, v_1\}$ in $E'$. Now $G'$ is bipartite, and we know the bipartition, with nodes $v_1$ on one side and nodes $v_2$ on the other side. We can now apply any algorithm that finds a matching $x'$ in the bipartite graph $G'$, and this can be mapped into a half-integral matching $x$ by setting
\begin{equation}\label{eq:intro-half}
    x[\{u,v\}] = \frac{x'[\{u_1,v_2\}] + x'[\{u_2,v_1\}]}{2}.
\end{equation}
Hence, we could use any distributed algorithm designed for bipartite graphs---there is a very simple algorithm that finds a maximal matching in bipartite graphs in $O(\Delta)$ rounds independently of $n$. Then by applying \eqref{eq:intro-half} we could turn it into a fractional matching.

There is, unfortunately, a catch: while \eqref{eq:intro-half} will preserve feasibility (given a matching $x'$ it will result in a fractional matching $x$), it will not preserve maximality: even if $x'$ is a maximal matching, it is not necessarily the case that $x$ is a maximal fractional matching. Could we nevertheless find a half-integral matching efficiently with a distributed algorithm?

If we consider prior distributed algorithms for maximal fractional matching \cite{Delta-squared,VC-SC}, they are very far from being able to produce half-integral matchings. For example, \cite{Delta-squared} uses fractional values with denominators as large as $2^{\Delta-1}$ and \cite{VC-SC} is even worse. In this work we show that denominators exponential in $\Delta$ are necessary, but we can still do better than prior work.

\subsection{Contributions}

Our main result is a full characterization of exactly how fine-grained fractional values are necessary, for any algorithm whose running time is independent of $n$, or at most $o(\log^* n)$ as a function of $n$:
\begin{theorem}[Upper bound]\label{thm:upper}
    Fix any constant $d$, and let $\Delta = 2d+1$. Then there is a distributed algorithm that runs in $5\Delta^3$ rounds and that finds a maximal fractional matching in graphs of maximum degree $\Delta$ using only fractional numbers of the form $a/b$ where $a = 0, 1, \dotsc, 2^d$ and $b = 2^d$.
\end{theorem}

\begin{theorem}[Lower bound]\label{thm:lower}
    Fix any constant $d$, and let $\Delta = 2d+2$. Then there is no distributed algorithm that runs in $o(\log^* n)$ rounds and that finds a maximal fractional matching in graphs of maximum degree $\Delta$ using only fractional numbers of the form $a/b$ where $a = 0, 1, \dotsc, 2^d$ and $b = 1, 2, \dotsc, 2^d$.
\end{theorem}

We emphasize that the lower bound result holds for any given $d$, and the constants hidden in the $o$-notation may depend on $d$. In particular, it holds for algorithms that run in $T(\Delta) + o(\log^* n)$ for any function $T$. This is also tight: in $O(\Delta + \log^* n)$ rounds we can find an integral solution \cite{MM3}. One surprise here is that there is no middle ground between $o(\log^* n)$-round algorithms that require exponentially fine-grained fractional values and $O(\log^* n)$-round algorithms that can use integral values.

We note that the upper bound only uses multiples of $1/2^d$, while the lower bound also excludes the possibility of finding a maximal matching using, e.g., values that are multiples of $1/d$.

The range of fractional numbers our algorithm uses is much smaller than those in the prior work. In our algorithm, the denominator is upper bounded by $2^{\Delta/2}$, while in prior work \cite{Delta-squared} it is approximately $2^{\Delta}$. See \cref{tab:comparison} for a comparison between different algorithms.

\begin{table}
    \centering
    \begin{tabular}{lll}
        \toprule
        Values & Round complexity & Reference \\
        \midrule
        multiples of $\frac{1}{\Delta!}$ & $O(\Delta)$ & \cite{VC-SC} \\[4pt]
        multiples of $\frac{1}{2^{\Delta-1}}$ & $O(\Delta^2)$ & \cite{Delta-squared} \\[4pt]
        multiples of $\frac{1}{2^{\lfloor \Delta/2 \rfloor}}$ & $O(\Delta^3)$ & this work \\[4pt]
        integral & $O(\Delta + \log^* n)$ & \cite{MM3} \\[2pt]
        \bottomrule
    \end{tabular}
    \caption{Trade-offs between fractional values and round complexity.}\label{tab:comparison}
\end{table}

As a corollary of our results, we also have a full characterization of the complexity of half-integral matchings:
\begin{corollary}\label{cor:upper-half}
    It is possible to find a maximal half-integral matching in graphs of maximum degree $\Delta = 3$ in $O(1)$ rounds.
\end{corollary}
\begin{corollary}\label{cor:lower-half}
    It is not possible to find a maximal half-integral matching in graphs of maximum degree $\Delta = 4$ in $o(\log^* n)$ rounds.
\end{corollary}

Again, this is tight, as one can find not just a half-integral but also an integral solution for any fixed $\Delta$ in $O(\log^* n)$ rounds.

\subsection{Key new ideas}

While the upper bound of \cref{thm:upper} is a relatively simple adaptation of ideas from prior work, the lower bound of \cref{thm:lower} requires a development of a new proof strategy.

Prior lower-bound techniques in this area tend to fall in one of these categories, each unsuitable for us:
\begin{enumerate}
    \item The lower-bound construction is a regular graph \cite{MM4,Greedy-Optimal-Distributed-Maximal-Matching}. In $\Delta$-regular graphs we can trivially find a fractional matching using the value $1/\Delta$, which is exponentially far from the lower bound in \cref{thm:lower} that we aim at proving.
    \item The lower-bound result aims at establishing that one needs some specific number of rounds, e.g., $\Omega(\Delta)$ rounds \cite{MM4,Linear-in-Delta,Greedy-Optimal-Distributed-Maximal-Matching}. However, in \cref{thm:lower} we aim at proving that even if the round complexity is, say, exponential in $\Delta$, one cannot avoid using fine-grained fractional values.
\end{enumerate}
Our proof strategy superficially resembles the one used in \cite{Linear-in-Delta,Greedy-Optimal-Distributed-Maximal-Matching} in the sense that we start with one node and $k$ self-loops, which represents the local view of a node in the middle of a regular graph, and then we start unfolding the loops. At each point of the process we see what is the output the algorithm commits to, and then we continue the process until we are left with a concrete lower-bound graph. However, there are major differences; see \cref{fig:comp-with-prior}:
\begin{figure}[t]
    \centering
    \includegraphics[page=2]{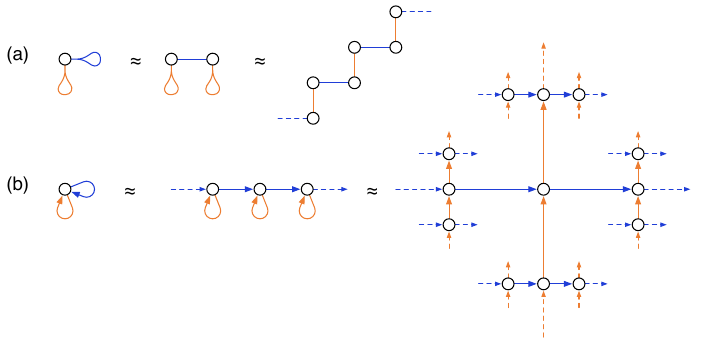}
    \caption{(a) In prior work \cite{Linear-in-Delta,Greedy-Optimal-Distributed-Maximal-Matching}, all the heavy lifting is done in a so-called EC model, in which edges are undirected but colored. Self-loops represent undirected edges. For example, a node with $2$ self-loops represents a node in the middle of a $2$-regular tree, i.e., a long path. (b) In this work, we work in the PO model. Self-loops represent long directed paths. For example, a node with $2$ self-loops represents a node in the middle of a $4$-regular tree in which all nodes have indegree $2$ and outdegree $2$.}\label{fig:comp-with-prior}
\end{figure}
\begin{itemize}
    \item In \cite{Linear-in-Delta,Greedy-Optimal-Distributed-Maximal-Matching} they start with a \emph{pair of nodes}. The nodes have self-loops, and each self-loop represents an \emph{undirected edge}; the entire argument relies on the fact that an algorithm cannot break symmetry between two ends of such an edge. At each step they unfold a relevant loop, doubling the number of nodes, and then they mix elements from two instances, resulting in another pair of instances. In each iteration they lose one self-loop but force the algorithm to look one step further.
    \item In this work we start with a \emph{single node}. The node has self-loops, but this time each self-loop represents a \emph{long directed path}; our argument relies on the fact that an algorithm cannot break local symmetry between two nodes near the middle of the path. At each step we unfold a relevant loop, but this will turn one node into a directed path of length $\Theta(T)$. We are interested in the behavior of the algorithm both in the middle of the path and at the endpoints of the path. In each iteration we lose one self-loop but force the algorithm to use at least twice as large denominators.
\end{itemize}

\section{Preliminaries}

\subsection{Graphs and self-loops}

For a graph $G=(V,E)$, we write $\Delta(G)$ to denote the maximum degree of the graph. We use just $\Delta$ when $G$ is clear from the context. For any natural number $d \in \mathbb{N}$, we use $\mathcal{G}_{d}$ to represent the family of graphs such that $G \in \mathcal{G}_{d}$ if $\Delta(G) \le d$. Throughout this work, we will assume that the maximum degree of the input graph $G$ is a globally known constant.

In what follows, we will refer to a self-loop simply as a loop. Each loop counts as one incoming and one outgoing edge (in particular, in $G_{2d}$ a node can have at most $d$ self-loops). We call a graph \emph{loopy} if each vertex of the graph has at least one loop.

\subsection{Model of computing}\label{ssec:prelim-model}

Our main results, \cref{thm:upper,thm:lower}, hold in the usual LOCAL model \cite{Linial-Seminal-Paper,Peleg-book}. For simplicity, we will focus here on deterministic algorithms (even though the results are not hard to extend to randomized algorithms).

However, to prove the lower bound result, it will be convenient to first prove the lower bound in a weaker model (called PO here, following \cite{Linear-in-Delta}) and then extend the result from the PO model to the LOCAL model. It will be easiest to define everything we need by starting with the deterministic port-numbering model (PN).

\paragraph{PN model (port numbering) \cite{angluin80local,yamashita96computing}.}

Let $G = (V,E)$ be the input graph. In the PN model, each node $v \in V$ is a computer and each edge $\{u,v\} \in E$ is a communication link between two computers. Initially, each computer is only aware of its degree; nodes of the same degree start with the same initial state.

The endpoints of the edges are labeled with \emph{port numbers}; a node of degree $d$ can refer to its incident edges with the numbers $1, 2, \dotsc, d$; see \cref{fig:models}. The port numbering comes from an adversary; a distributed algorithm in the PN model has to work correctly for any given port numbering.

Computation proceeds in \emph{synchronous communication rounds}. In each round, each node can
\begin{enumerate}[noitemsep]
    \item send a message to each neighbor,
    \item receive a message from each neighbor, and
    \item update its local state based on the current state and the messages it received.
\end{enumerate}
After each round, each node can decide whether it stops and announces its own part of the output---in the case of the maximal fractional problem, the output of a node indicates the fractional value assigned to each incident edge. The \emph{running time} of the algorithm is the number of rounds until all nodes have stopped and announced their local outputs.

\paragraph{PO model (port numbering and orientation) \cite{mayer95local,Linear-in-Delta}.}

Algorithms in the PO model behave in exactly the same way as in the PN model. However, there is one additional piece of information available to the algorithm: each edge $\{u,v\} \in E$ is oriented (arbitrarily, by the adversary); see \cref{fig:models}. More precisely, each node knows for each incident edge whether it is ``outgoing'' or ``incoming''.

While an arbitrary orientation may not seem particularly useful, note that the PO model is strictly stronger than the PN model. For example, if we have a graph $G$ with two nodes and one edge, it is trivial to find a proper $2$-coloring of $G$ in the PO model in $0$ rounds, while it is impossible to solve in the PN model in any number of rounds.

\paragraph{LOCAL model \cite{Linial-Seminal-Paper,Peleg-book}.}

Algorithms in the LOCAL model also behave in exactly the same way as in the PN model, but there is again one additional piece of information available to the algorithm: each node is labeled (arbitrarily, by the adversary) with a \emph{unique identifier} from a polynomially-sized set; see \cref{fig:models}.

Again, the LOCAL model is strictly stronger than the PO model. For example, maximal matching cannot be found in the PO model if the input graph is a cycle that is consistently oriented, while the task is solvable in the LOCAL model in $O(\log^* n)$ rounds.

However, it turns out that \emph{constant-time} algorithms in the LOCAL model are not much stronger than algorithms in the PO model, see e.g.\ \cite{Linear-in-Delta,goos13local-approximation}. This is the idea we will also make use of in this work: our main goal is to prove a lower bound in the LOCAL model, but it will be convenient to first study the PO model.

\begin{figure}[h!]
    \centering
    \includegraphics[page=3]{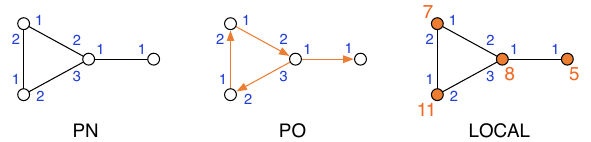}
    \caption{Models of computing used in this work.}\label{fig:models}
\end{figure}

\subsection{Applying PO algorithms to loopy graphs} \label{ssec:PO-loopy}

To prove the lower-bound result of \cref{thm:lower}, we will study the output of a PO algorithm $\mathcal{A}$ in some loopy graph $G$. However, when we consider distributed graph algorithms, we usually assume that the input graph is loop-free.

However, the output of $\mathcal{A}$ in loopy graphs is nevertheless well-defined. When we refer to the output of $\mathcal{A}$ on some edge $e$ in $G$, we refer to the result of the following thought experiment: Unfold all loops in $G$, as shown in \cref{fig:comp-with-prior}b, and hence we arrive at a tree $G'$. Then apply $\mathcal{A}$ in $G'$ (as the running time of $\mathcal{A}$ is independent of the size of the input graph, this is well-defined). Edge $e$ in $G$ corresponds to infinitely many edges $e'$ in $G'$, but each such edge is symmetric and hence the output of $\mathcal{A}$ on each such edge $e'$ is the same; hence we can take any such edge $e'$ and interpret its label as the output of $\mathcal{A}$ on $e$.

In particular, if $\mathcal{A}$ finds a maximal fractional matching in any loop-free graph $G'$, it will also produce a maximal fractional matching in the loopy graph $G$ (the label of the loop is counted twice).

\section{Lower bound result}

In this section we prove the lower-bound result, \cref{thm:lower}. It turns out that the critical resource is the number of factors of $2$ in the denominators. We start by defining sets of rational numbers that will precisely capture how fine-grained values are needed.

\subsection{Sets of rational numbers}

Any natural number $x \ge 1$ can be written as $x = 2^n \cdot m$ where $n \ge 0$ and $m \equiv 1 \mod{2}$. We refer to $e(x) = 2^n$ as the \emph{even part} of $x$ and $o(x) = m$ as the \emph{odd part} of $x$. For $x = 0$, we define $e(x) = 0$ and $o(x) = 1$.

We extend this notion to rational numbers as follows. If $x = p/q$ in the reduced form, we define the \emph{even part of the denominator} $\bar e(x) = e(q)$ and the \emph{odd part of the denominator} $\bar o(x) = o(q)$. For example, $\bar e(0/1) = \bar e(1/1) = 1$, $\bar e(1/3) = 1$ and $\bar e(1/4) = 4$.

For each $n\ge 1$, we define
\begin{align*}
    R_n &= \bigl\{ x \in \mathbb{Q} : 0 \le x \le 1 \text{ and } \bar e(x) = 2^n \bigr\}, \\
    R_{ \le n} &= R_0 \cup R_1 \cup \dotsb \cup R_n, \\
    R_{\ge n} &= R_n \cup R_{n+1} \cup \dotsb, \\
    R_{> n} &= R_{n+1} \cup R_{n+2} \cup \dotsb.
\end{align*}
For example, we have
\begin{align*}
    R_0 &= \bigl\{ \textstyle 0, 1, \frac{1}{3}, \frac{2}{3}, \frac{1}{5}, \frac{2}{5}, \frac{3}{5}, \frac{4}{5}, \dotsc \bigr\}, \\
    R_1 &= \bigl\{ \textstyle \frac{1}{2}, \frac{1}{6}, \frac{5}{6}, \dotsc \bigr\}, \\
    R_2 &= \bigl\{ \textstyle \frac{1}{4}, \frac{3}{4}, \frac{1}{12}, \frac{5}{12}, \frac{7}{12}, \frac{11}{12}, \dotsc \bigr\}.
\end{align*}
We can view $R_n$ as the set of fractional number whose denominator has exactly $n$ trailing zeros in its binary representation. Note that for each rational number $x \in [0,1]$ there exists exactly one $n$ such that $x \in R_n$. For $m < n$, we have $R_{\le m} \subsetneq R_{\le n}$.

\subsection{High-level plan}

In \cref{ssec:lower-local} we prove the following lemma, which essentially shows that we can without loss of generality focus on the PO model:
\begin{lemma}\label{lem:lower-local}
Fix a natural number $\Delta \in \mathbb{N}$. Then, for any natural number $T \in \mathbb{N}$, the following holds: if there exists a $T$-round algorithm that solves the maximal fractional matching problem using values in a set $\mathscr{R}$ in the LOCAL model on trees with maximum degree $\Delta$, then there exists a $T$-round algorithm that solves the maximal fractional matching problem using values in set $\mathscr{R}$ in the PO model for any loopy graph $G$ with maximum degree $\Delta$. 
\end{lemma}

Then in \cref{ssec:lower-pn} we prove the following lemma, which captures exactly how fine-grained rational values are needed in the PO model:
\begin{lemma}\label{lem:lower-pn}
Fix natural number $d \in \mathbb{N}$. Then, for any natural number $T \in \mathbb{N}$, there does not exist any algorithm in the PO model that uses $T$ rounds and computes a valid solution for the maximal fractional matching problem using the values from $R_{ \le (d-1)}$ for loopy graphs in graph family $\mathcal{G}_{2d}$. 
\end{lemma}

By putting together \cref{lem:lower-local} and \cref{lem:lower-pn}, we obtain:
\begin{lemma}\label{lem:lower}
Fix a natural number $d \in \mathbb{N}$. Then, for any natural number $T \in \mathbb{N}$, there does not exist any algorithm in the LOCAL model that uses $T$ rounds and computes a valid solution for the maximal fractional matching problem using the values from $R_{ \le (d-1)}$ for any tree in the graph family $\mathcal{G}_{2d}$. 
\end{lemma}

Now we are almost done; we just need to amplify this from $T = O(1)$ to $T = o(\log^* n)$. To this end, there are two key observations. The first observation is that \cref{lem:lower-pn} holds even when we use values from the set  $S = \{ a/b : a = 0, 1, \dotsc, 2^{d-1}$ and $b = 1, 2, \dotsc, 2^{d-1}\}$ instead of $R_{\le d-1}$, as $S$ is a finite subset of $R_{\le d-1}$. The second observation is that maximal fractional matching using values from $S$ in the graph family $\mathcal{G}_{2d}$ is a \emph{locally checkable labeling} problem (LCL), as defined by Naor and Stockmeyer \cite{naor-stockmeyer}. The recent result by Grunau, Rozhon, and Brandt \cite{LCL-Trees-Landscape} shows that an $o(\log^* n)$-round algorithm for any LCL problem $\Pi$ in trees implies an $O(1)$-round algorithm for the same problem $\Pi$. Hence we obtain as a corollary of \cref{lem:lower}:
\begin{corollary}\label{cor:lower}
    Fix a natural number $d \in \mathbb{N}$. Then there does not exist any algorithm in the LOCAL model that uses $o(\log^* n)$ rounds and computes a valid solution for the maximal fractional matching problem for any tree in the graph family $\mathcal{G}_{2d}$ using values of the form $a/b$ where $a = 0, 1, \dotsc, 2^{d-1}$ and $b = 1, 2, \dotsc, 2^{d-1}$. 
\end{corollary}
And naturally the problem does not get any easier if we consider general graphs instead of trees. Our main lower bound result, \cref{thm:lower}, directly follows.

\subsection{Proof of Lemma \ref{lem:lower-local}}\label{ssec:lower-local}

In \cite{Linear-in-Delta}, a similar result is shown with two exceptions: the graph family is not restricted to trees and the edge values are not restricted to $\mathscr{R}$. However, the same proof follows when we add these restrictions. This result is a simple extension of \cite[Sections 5.3--5.4]{Linear-in-Delta}, where we can see that the simulation argument only simulates the LOCAL algorithm on certain types of neighborhoods (and these neighborhoods are actually trees) and it does not make changes in the value used for the PO model.

\subsection{Proof of Lemma \ref{lem:lower-pn}}\label{ssec:lower-pn}

\paragraph{Preliminary Observations.}

We first make a few observations regarding our problem. First recall the way in which we use loops to represent a node in the middle of a directed path (\cref{fig:comp-with-prior}).
\begin{observation}\label{obs:loop-satur}
    If a node has a loop then it must be saturated.
\end{observation}
\begin{proof}
    If a node with a loop was not saturated, we would have a directed path of unsaturated nodes and, in particular, edges with unsaturated endpoints.
\end{proof}

In a saturated node, the labels of incident edges have to sum up to $1$. The following observation captures a key property related to how the even parts of the denominators behave when rational numbers sum up to $1$.
\begin{observation}\label{obs:even-denom}
    Let $n \ge 1$ and $\frac{k}{m \cdot 2^n} \in R_n$.
    Consider the equation
    \[
        2 \ell_1 + \ldots + 2 \ell_r + x_1 + \ldots + x_{r'} + \frac{k}{m \cdot 2^n} = 1,
    \]
    where each $\ell_i$ and $x_i$ can be any non-negative rational number. Then, either $\ell_i \in R_{> n}$ or $x_i \in R_{\ge n}$ for some $i$. Put otherwise, either some $\ell_i$ has the even part of the denominator larger than $2^n$ or some $x_i$ has the even part of the denominator at least $2^n$. 
\end{observation}
\begin{proof}
First consider the equation
\[
x_1 + \ldots + x_q +\frac{k}{m \cdot 2^n} = 1
\]
in which each $x_i$ can be any non-negative rational number. We show that there exists an index $i$ for which $x_i \in R_{\ge n}$. We can rewrite it as solving the equation
\[
x_1 + \ldots + x_q =\frac{m\cdot 2^n - k}{m \cdot 2^n},
\]
where $\frac{m\cdot 2^n - k}{m \cdot 2^n} \in R_n $. If each $x_i$ had the even part of the denominator less than $2^n$, then $x_1 + \ldots + x_q$ would also have the even part of the denominator less than $2^n$. This is because when we add two rationals $\frac{a_1}{b_1}$ and $\frac{a_2}{b_2}$ we get  $$ \frac{a_1}{b_1} + \frac{a_2}{b_2} = \frac{a_1 \cdot (\ell / b_1) + a_2 \cdot (\ell / b_2)}{\ell}$$ where $\ell = \lcm(b_1,b_2)$, the least common multiple of $b_1$ and $b_2$. The even part of $\ell$ will be bounded above by the maximum of the even parts of $b_1$ and $b_2$. However, if $x_1 + \ldots + x_q$ has the even part of the denominator less than $2^n$, then it contradicts the fact that the sum equals $\frac{m\cdot 2^n - k}{m \cdot 2^n}$.

Now, in order to prove the original statement of \cref{obs:even-denom}, it is sufficient to replace $x_{r'+i}$ by $2 \ell_i$. If $x_{r'+i} \in R_{\ge n}$ then $\ell_i \in R_{> n}$. 
\end{proof}

\paragraph{Assumptions.}

We now proceed to prove \cref{lem:lower-pn} by contradiction. For the sake of contradiction we assume that when we fix a natural number $d \in \mathbb{N}$, there exists a natural number $T \in \mathbb{N}$ such that the following holds: there exists a PO algorithm $\mathcal{A}$ that solves the maximal fractional matching problem in $T$ rounds using values from the set $R_{ \le (d-1)}$ for graph family $\mathcal{G}_{2d}$.

\paragraph{Properties.}

Now, our lower bound construction observes the behavior of $\mathcal{A}$ on different kinds of graphs in $\mathcal{G}_{2d}$ to reason about the set of values that is used. We will construct a sequence of loopy graphs $G_0, G_1, \ldots, G_{d-1}$ to argue that the further we go, the more fine-grained value must be used by our algorithm.

For each $i = 0,1, \ldots d-1$, we will maintain the following properties:
\begin{enumerate}[label={\bfseries P\arabic*},leftmargin=*]
    \item $G_i \in \mathcal{G}_{2d}$.
    \item Graph $G_i$ without loops forms a tree.
    \item Each node of $G_i$ has at least $d-i$ loops.
    \item There is an integer $j(i) > i$ and a node $v_i$ in $G_i$ such that $\mathcal{A}$ labels at least one loop of $v_i$ with a rational value $x \in R_{j(i)}$.
\end{enumerate}

\paragraph{Base Case.}

Our first graph $G_0$ consists of a single node $v_0$ with $d$ oriented self loops (see \cref{fig:Gi}).

Graph $G_0$ satisfies properties \textbf{P1}, \textbf{P2} and \textbf{P3} by construction, so we now need to verify only \textbf{P4}. Consider that $\mathcal{A}$ assigns values $a_1, \ldots, a_d$ to the loops of~$v_0$. Since $v_0$ has loops, it must be saturated (recall \cref{obs:loop-satur}), and hence it must satisfy that $2 a_1 + 2 a_2 + \ldots + 2 a_d = 1$. This is equivalent to solving $a_1 + a_2 + \ldots + a_d = 1/2$ and by \cref{obs:even-denom} we know that there exists an $i$ with $a_i \in R_{\ge 1}$.

\begin{figure}[t]
    \centering
    \includegraphics[page=4]{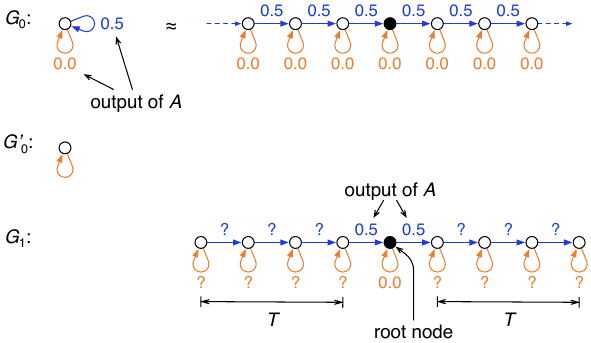}
    \caption{Construction for $d=2$ and $T=3$. Graph $G_0$ consists of $d$ self-loops. When we apply $\mathcal{A}$ to $G_0$, at least one of the loops will get labeled by a value in $R_{\ge 1}$; in this example the value was $0.5 \in R_1$. To construct $G_1$, we remove this loop to arrive at graph $G'_0$, take $2T+3$ copies of $G'_0$, and connect them with a directed path. The key observation is that given the output of $\mathcal{A}$ in $G_0$ we also know the output of $\mathcal{A}$ around the node in the middle of $G_1$---this node is called the \emph{root node} of $G_1$.}\label{fig:Gi}
\end{figure}
    
\paragraph{Inductive Step.}

Given $G_{i-1}$, we construct $G_i$ as follows; see \cref{fig:Gi}:
\begin{enumerate} [label={\bfseries S\arabic*},leftmargin=*]
    \item Construct the graph $G'_{i-1}$ from $G_{i-1}$ by removing the loop of $v_{i-1}$ for which $\mathcal{A}$ assigned a value in $R_{j(i-1)}$. 
    \item Create $2T+3$ copies of $G'_{i-1}$.
    \item For each $k = 1, 2, \dots, 2T+2$, connect node $v_{i-1}$ in copy number $k$ to node $v_{i-1}$ in copy number $k+1$; these new edges are called \emph{path edges}.
    \item Node $v_{i-1}$ in copy number $T+2$ is called the \emph{root node} of $G_i$.
\end{enumerate}

This way we form a directed path of length $2T+3$, with the root node in the middle of the path, as shown in \cref{fig:Gi}. The key observation is that the output of algorithm $\mathcal{A}$ on the root node of $G_{i}$ is the same as the output of $\mathcal{A}$ for $v_{i-1}$ in $G_{i-1}$, due to the fact that the radius-$T$ neighborhood of the root node in $G_i$ is isomorphic to the radius-$T$ neighborhoods of $v_{i-1}$ in $G_{i-1}$ (once we conceptually unfold all loops). This property is illustrated in \cref{fig:Gi}: compare the radius-$T$ neighborhood of the black node in the unfolding of $G_0$ with the radius-$T$ neighborhoods of the root node of $G_1$.

Given $G_{i-1}$ satisfies all the properties, we need to show that the same is true for $G_i$. \textbf{P1}, \textbf{P2} and \textbf{P3} are satisfied by construction. To prove \textbf{P4}, consider the root node of $G_i$. Since its behavior is completely characterized, we know that it will label the incident path edges with values from $R_{j(i-1)}$.

Recall that by \textbf{P2} graph $G_i$ without loops forms a tree. We will navigate in this tree, starting from the root node, and moving away from it until we satisfy \textbf{P4}. We maintain the following invariant; see \cref{fig:Gi-ind}:
\begin{definition}[path invariant]
    If $v$ is the current node, and $P$ is the unique path from $v$ to the root, we have already concluded that $\mathcal{A}$ labels each edge of $P$ with a value from $R_{\ge j(i-1)}$.
\end{definition}
To get started, let $e$ be one of the path edges incident to the root node, and let $v$ be the other end of $e$. As we discussed earlier, we know that $e$ is labeled with a value from $R_{j(i-1)}$.

\begin{figure}
    \centering
    \includegraphics[page=5]{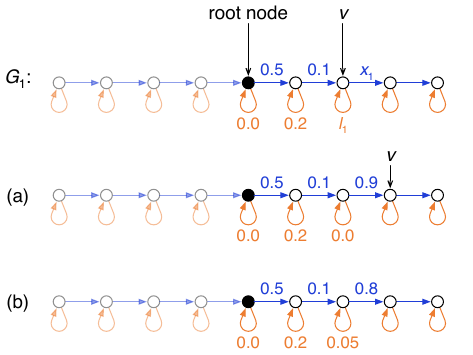}
    \caption{Inductive step in the proof of \cref{lem:lower-pn} (\cref{ssec:lower-pn}). We have already concluded that all edges in the path between $v$ and the root node are labeled with values from $R_{\ge 1}$. We now ask how algorithm $\mathcal{A}$ will label the other edges around $v$. (a)~One possible solution: edge $x_1$ is labeled with a value $0.9 = \frac{9}{2 \cdot 5} \in R_1$. We did not yet establish property \textbf{P4}, but we can extend the $R_1$-labeled path further away from the root node---eventually we will encounter a leaf node. (b)~Another possible solution: we managed to label $x_1$ with a less fine-grained value $0.8 \in R_0$. However, this means that loop $\ell_1$ is labeled with a more fine-grained value $0.05 = \frac{1}{2^2 \cdot 5} \in R_2$. We have established \textbf{P4}.}\label{fig:Gi-ind}
\end{figure}

Now assume that we have reached some node $v$ this way. Let $P$ be the path from $v$ to the root, and let $e$ be the first edge of $P$, let $L$ be the set of loops incident to $v$, and let $X$ be the set of non-loop edges incident to $v$ that are different from $e$. That is, we already know the label of edge $e$, but we do not yet know how $\mathcal{A}$ will label $L$ and $X$.

Node $v$ is loopy, so it must be saturated. The saturation condition for $v$ is equivalent to solving the equation \[2\ell_1 + \ldots + 2 \ell_r + x_1 + \ldots + x_{r'} + \frac{k}{m \cdot 2^n}  = 1,\] where $n \ge  j(i-1)$, values $\ell_i$ represent the values assigned to the loops in $L$, values $x_i$ represent the values assigned to the edges in $X$, and $\frac{k}{m \cdot 2^n}$ refers to the value from $R_{ \ge j(i-1)}$ assigned to edge $e$. With the help of \cref{obs:even-denom}, we know that one of the two cases must be true:
\begin{enumerate}
    \item One of the loops in $L$ has the even part of the denominator $2^{n'}$ for $n' > n$. In this case, we have established \textbf{P4}.
    \item One of the edges $\{u,v\} \in X$ has the even part of the denominator at least $2^{n}$. We have found another edge labeled with a value from $R_{\ge j(i-1)}$, and we can extend the path $P$ by moving from $v$ to $u$, still satisfying the path invariant.
\end{enumerate}
Note that this process will eventually terminate, as $G_i$ without loops is a (finite) tree, and hence we will eventually reach a leaf node with $X = \emptyset$. We have established that our construction of graph $G_i$ satisfies properties \textbf{P1}--\textbf{P4}.

\paragraph{Conclusion.}

When we take $i=d-1$, we have a graph $G_{d-1} \in \mathcal{G}_{2d}$ which requires the even part of the denominator to be at least $2^{d}$. However, values with denominator $2^d$ are not present in the set $R_{ \le (d-1)}$. Thus, we have our desired contradiction.

This concludes the proof of \cref{lem:lower-pn}, and hence also the proofs of \cref{lem:lower} and our main lower bound result \cref{thm:lower}.

\section{Upper bound result}

Here, we prove the statement of \cref{thm:upper}. We will use the notation
\[
    S(d) = \Bigl\{ \frac{i}{2^d} : i \in \{0,1, \ldots,2^d \} \Bigr\}.
\]
We need to show that there is a $5\Delta^3$-round distributed algorithm that solves maximal fractional matching in graph family $\mathcal{G}_{2d+1}$ using labels from $S(d)$. We prove the claim by induction, as follows:
\begin{itemize}[noitemsep]
    \item Base case (\cref{lem:upper-base-case}): $S(1)$ suffices for $\mathcal{G}_{2}$.
    \item Odd step (\cref{lem:upper-odd-case}): if $S(d)$ suffices for $\mathcal{G}_{2d}$, then $S(d)$ also suffices for $\mathcal{G}_{2d+1}$.
    \item Even step (\cref{lem:upper-even-case}): if $S(d)$ suffices for $\mathcal{G}_{2d+1}$, then $S(d+1)$ suffices for $\mathcal{G}_{2d+2}$.
\end{itemize}
We will give a PN algorithm, which implies the existence of a LOCAL algorithm.

\begin{lemma}\label{lem:upper-base-case}
    There is a $1$-round PN algorithm that finds a maximal fractional matching in $\mathcal{G}_2$ using values from $S(1)$.
\end{lemma}

\begin{proof}
In this case, we want to pick $x(e) \in \{ 0, \frac{1}{2}, 1 \}$ for each $e \in E$. We can achieve a simple distributed algorithm with $1$ round of communication. Each vertex $v$, communicates its degree to its neighbors. Any degree $2$ vertex can safely assign the value $\frac{1}{2}$ to both of its incident edges. For a degree $1$ vertex, it will assign the value $\frac{1}{2}$ to the incident edge if the other endpoint has degree $2$ and will assign the value $1$, if the other endpoint is $1$ as well. 

We can see that for each vertex $v$, the sum of the values assigned to its incident edges is at most $1$. By the nature of our algorithm, every degree $2$ node is saturated. So, every edge which has a degree $2$ endpoint satisfies that one of its endpoints is saturated. The only remaining scenario is when both of the endpoints are degree $1$. In this setting, our algorithm assigns the edge with value $1$ in which case both of its endpoints are saturated as well. Using $1$ round of communication, we have obtained a solution for the maximal fractional matching using values $\{ 0, \frac{1}{2}, 1 \}$ when $\Delta=2$.
\end{proof}

Let us now rephrase \cref{lem:upper-base-case} as follows so that we have a clean starting point for the induction:

\begin{corollary}\label{cor:upper-base-case}
For $d = 1$, there is a ${5\cdot(2d)^3}$-round PN algorithm that finds a maximal fractional matching in $\mathcal{G}_{2d}$ using values from $S(d)$.
\end{corollary}

Now we are ready to take the step from degree $2d$ to degree $2d+1$:

\begin{lemma}\label{lem:upper-odd-case}
Fix $d \in \mathbb{N}$. Assume that there is a ${5\cdot(2d)^3}$-round PN algorithm that finds a maximal fractional matching in $\mathcal{G}_{2d}$ using values from $S(d)$.
Then there is a ${5\cdot(2d+1)^3}$-round PN algorithm that finds a maximal fractional matching in $\mathcal{G}_{2d+1}$ using values from $S(d)$.
\end{lemma}

\begin{proof}
Assume that $\mathcal{A}$ is a PN algorithm that computes the solution for $\mathcal{G}_{2d}$ using values in $S(d)$. We now describe PN algorithm $\mathcal{A}'$ that computes the solution for $G \in \mathcal{G}_{2d + 1}$ using values in $S(d)$. Algorithm $\mathcal{A}'$ takes the following steps (see \cref{fig:example} for an illustration):
\begin{figure}[t]
    \centering
    \includegraphics[page=6]{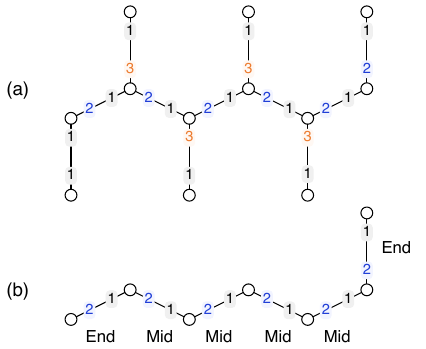}
    \caption{(a) A graph $G \in \mathcal{G}_{3}$, with a port numbering. (b) The subgraph $G_\ell$ for label $\ell = \{1,2\}$, with the edge types ``End'' and ``Mid'' indicated.}\label{fig:example}
\end{figure}
\begin{description}
    \item[Step 1: Edge labelling.]
    First, we use the port numbers to define a \emph{label} for each edge. For each edge $e = \{u,v\}$, there exists numbers $i,j \le \Delta(G)$ such that port $i$ of $u$ is connected to port $j$ of $v$. We label this edge with the set $\{i,j\}$. Then $L = \{  \{i,j\} : 1 \le i , j \le \Delta(G) \}$ denotes the set of possible edge labels. We have $|L| = \Delta^2$ different edge labels. For each $\ell \in L$, we define the subgraph $G_{\ell}$ of $G$ that contains all the edges labelled $\ell$. We write $\deg_{G_{\ell}}(v)$ for the degree of node $v$ in graph $G_\ell$. A key observation is that for each $\ell$ and $v$, we have $\deg_{G_{\ell}}(v) \le 2$, i.e., each $G_{\ell}$ is a collection of paths and cycles. Edge labelling can be done in $1$ communication round as a node $v$ just needs to communicate the relevant port for each of its incident edges.
    
    \item[Step 2: Edge Classification.]
    We classify each edge into two \emph{types}: ``Mid'' and ``End''. Consider any edge $e = \{u,v \}$ and say it had label $\ell \in L$. We say that $e$ is of type ``Mid'' if $\deg_{G_{\ell}}(u)=2$ and $\deg_{G_{\ell}}(v) = 2$. Put otherwise, all edges that are in the middle of the path or part of a cycle in $G_{\ell}$ are classified with type ``Mid''. All other edges are classified as ``End''. After step 1, this can be done in $1$ round of communication as each node just needs to communicate its degree for different edge labels. 
    
    \item[Step 3: Solve for ``Mid'' edges.]
    Consider subgraph $G'$ of $G$ that contains all edges of type ``Mid''. We argue that $\Delta(G') \le 2d$. To see this, consider any vertex $v \in V$. If $\deg_G(v) = 2d+1$, there exists $\ell \in L$ such that $\deg_{G_{\ell}}(v) = 1$, and therefore at least one edge adjacent to $v$ will receive type ``End'' and will not be part of $G'$. Now we have a subgraph $G'$ of $G$ with $\Delta(G') \le 2d$, and we can simulate $\mathcal{A}$ in $G'$. By assumption, this step takes ${5\cdot(2d)^3}$ rounds. 
    
    \item[Step 4: Extend for ``End'' edges.]
    We notice that each edge $e \in G'$ satisfies the maximality condition, i.e., at least one endpoint is saturated. Thus, we now need to ensure the same for edges of type ``End''. For a label $\ell \in L$, let $G_{\ell}^{\textrm{End}}$ be the set of edges labelled $\ell$ of type ``End''. We know that edges of type ``End'' can only be part of paths of length $1$ and $2$ in $G_{\ell}^{\textrm{End}}$. We proceed to satisfy the maximality condition for edges of type ``End'' by considering them sequentially on the labels $\ell \in L$. 
    Consider an edge $e = \{u,v\} \in G_{\ell}^{\textrm{End}}$. If we assign $x(e) = \min\{ 1- x[u] , 1 - x[v]\}$ then we can ensure that $e$ satisfies the maximality condition along with ensuring that both $u$ and $v$ satisfy the feasibility condition. The only issue that can arise here is that some other edge adjacent to $u$ or $v$ is trying to update its value in parallel with edge $e$. Since we are looking at edge of type ``End'' and proceeding sequentially based on label $\ell \in L$, the above issue can only be caused by paths of length $2$. However, the middle vertex of this path can decide the sequential order in which the two edges are considered, after which this issue is avoided. For each label $\ell$, this can be done in two communication rounds. In the first round, each vertex of unassigned edges just communicate its degree and saturation level to the neighbors. If we have a path of length $2$, the middle node will decide the values for both edges and communicate it in next round. Overall, this step takes $2 \cdot (2d+1)^2$ rounds to solve for all labels. 
\end{description}

Overall, the round complexity is
\[
    1 + 1 + 5\cdot(2d)^3 + 2\cdot(2d+1)^2 \le 5\cdot(2d+1)^3,
\]
as desired.
\end{proof}

\begin{lemma}\label{lem:upper-even-case}
Fix $d \in \mathbb{N}$. Assume that there is a ${5\cdot(2d+1)^3}$-round PN algorithm that finds a maximal fractional matching in $\mathcal{G}_{2d+1}$ using values from $S(d)$.
Then there is a ${5\cdot(2d+2)^3}$-round PN algorithm that finds a maximal fractional matching in $\mathcal{G}_{2d+2}$ using values from $S(d+1)$.
\end{lemma}

\begin{proof}
The proof for this theorem uses the same ideas as done in \cite{Delta-squared}. Consider any graph $G \in \mathcal{G}_{2d+2}$ and let $\mathcal{A}$ be the PN algorithm that uses values in $S(d)$ to compute a valid solution for graphs in $\mathcal{G}_{2d+1}$. We make use of the following definitions from \cite{Delta-squared}:

\begin{definition}[almost-saturating solutions]
A half-integral fractional matching $x\colon E \to \{0,\frac{1}{2},1\}$ is almost-saturating if the following conditions hold for each node $v$:
\begin{itemize}[noitemsep]
    \item If $x[v] = 0$, then $x[u] = 1$ for all neighbors $u$ of $v$.
    \item If $x[v] = 1/2$, then $x[u] = 1$ for at least one neighbor of $v$.
\end{itemize}
\end{definition}

\begin{definition}[half-saturated edges]
Consider an almost-saturating solution $x\colon E \to \{0,\frac{1}{2},1\}$. An edge $e= \{u,v\}$ is:
\begin{itemize}[noitemsep]
    \item half-saturated if $x[u] = x[v] = 1/2$,
    \item fully-saturated if $x[u]=1$ or $x[v] = 1$.
\end{itemize}
\end{definition}
In \cite{Delta-squared} there is an algorithm that finds an almost-saturating solution in $2\Delta+1 = {2\cdot(d+2)+1}$ rounds. Let $\bar{x}$ denote the almost-saturating solution for $G$, and we let $G'$ to be the subgraph induced by the half-saturated edges; note that for each node $v$ there has to be at least one incident edge that is not half-saturated. Hence $G' \in \mathcal{G}_{2d+1}$, and we can apply $\mathcal{A}$ to produce a solution $x'$ for $G'$ using values in set $S(d)$. We can then extend domain of $x'$ to $E$ by setting $x'(e) = 0$ for $e \not\in G'$. Setting $x(e) = \bar{x}(e) + x'(e)/2$ now gives a maximal fractional matching for the graph $G$. This is because for any edge $e = \{u,v\}$ in $G'$, we have $\bar{x}[u] = \bar{x}[v] = 1/2$ and $x'[u]=1$ or $x'[v] = 1$. Moreover, $x(e) \in S(d+1)$. 

The round complexity is
\[
    5\cdot(2d+1)^3 + 2\cdot(d+2)+1 \le 5\cdot(2d+2)^3,
\]
as desired.
\end{proof}

\cref{thm:upper} now follows by putting together \cref{cor:upper-base-case,lem:upper-odd-case,lem:upper-even-case}.

\section{Conclusions}

Our results give a complete characterization of how fine-grained fractional values are needed in a distributed algorithm that finds a maximal fractional matching in any running time $T(\Delta)$ that only depends on the maximum degree $\Delta$ and is independent of $n$, or is at most $o(\log^* n)$ as a function of $n$. The main open question is if we can achieve this bound in time $T(\Delta) = O(\Delta)$---this would be optimal by \cite{Linear-in-Delta}.

\section*{Acknowledgements}

This work was supported in part by the Academy of Finland, Grant 333837. We would like to thank the anonymous reviewers for their helpful feedback, and the members of Aalto Distributed Algorithms group for discussions. This is an extended version of a preliminary conference report \cite{dahal23half-integral}.

\bibliographystyle{plainurl}
\bibliography{ref}

\begin{thebibliography}{10}

\bibitem{angluin80local}
Dana Angluin.
\newblock Local and global properties in networks of processors (extended
  abstract).
\newblock In Raymond~E. Miller, Seymour Ginsburg, Walter~A. Burkhard, and
  Richard~J. Lipton, editors, {\em Proceedings of the 12th Annual {ACM}
  Symposium on Theory of Computing, April 28-30, 1980, Los Angeles, California,
  {USA}}, pages 82--93. {ACM}, 1980.
\newblock \href {https://doi.org/10.1145/800141.804655}
  {\path{doi:10.1145/800141.804655}}.

\bibitem{Delta-squared}
Matti {\AA}strand, Patrik Flor{\'{e}}en, Valentin Polishchuk, Joel Rybicki,
  Jukka Suomela, and Jara Uitto.
\newblock A local 2-approximation algorithm for the vertex cover problem.
\newblock In Idit Keidar, editor, {\em Distributed Computing, 23rd
  International Symposium, {DISC} 2009, Elche, Spain, September 23-25, 2009.
  Proceedings}, volume 5805 of {\em Lecture Notes in Computer Science}, pages
  191--205. Springer, 2009.
\newblock \href {https://doi.org/10.1007/978-3-642-04355-0_21}
  {\path{doi:10.1007/978-3-642-04355-0_21}}.

\bibitem{VC-SC}
Matti {\AA}strand and Jukka Suomela.
\newblock Fast distributed approximation algorithms for vertex cover and set
  cover in anonymous networks.
\newblock In Friedhelm~Meyer auf~der Heide and Cynthia~A. Phillips, editors,
  {\em {SPAA} 2010: Proceedings of the 22nd Annual {ACM} Symposium on
  Parallelism in Algorithms and Architectures, Thira, Santorini, Greece, June
  13-15, 2010}, pages 294--302. {ACM}, 2010.
\newblock \href {https://doi.org/10.1145/1810479.1810533}
  {\path{doi:10.1145/1810479.1810533}}.

\bibitem{MM4}
Alkida Balliu, Sebastian Brandt, Juho Hirvonen, Dennis Olivetti, Mika{\"{e}}l
  Rabie, and Jukka Suomela.
\newblock Lower bounds for maximal matchings and maximal independent sets.
\newblock {\em J. {ACM}}, 68(5):39:1--39:30, 2021.
\newblock \href {https://doi.org/10.1145/3461458} {\path{doi:10.1145/3461458}}.

\bibitem{VC-Approx-Linear-Time}
Reuven Bar{-}Yehuda and Shimon Even.
\newblock A linear-time approximation algorithm for the weighted vertex cover
  problem.
\newblock {\em J. Algorithms}, 2(2):198--203, 1981.
\newblock \href {https://doi.org/10.1016/0196-6774(81)90020-1}
  {\path{doi:10.1016/0196-6774(81)90020-1}}.

\bibitem{Conference-Locality-of-Symmetry-Breaking}
Leonid Barenboim, Michael Elkin, Seth Pettie, and Johannes Schneider.
\newblock The locality of distributed symmetry breaking.
\newblock In {\em 53rd Annual {IEEE} Symposium on Foundations of Computer
  Science, {FOCS} 2012, New Brunswick, NJ, USA, October 20-23, 2012}, pages
  321--330. {IEEE} Computer Society, 2012.
\newblock \href {https://doi.org/10.1109/FOCS.2012.60}
  {\path{doi:10.1109/FOCS.2012.60}}.

\bibitem{ACM-Journal-Locality-of-Symmetry-Breaking}
Leonid Barenboim, Michael Elkin, Seth Pettie, and Johannes Schneider.
\newblock The locality of distributed symmetry breaking.
\newblock {\em J. {ACM}}, 63(3):20:1--20:45, 2016.
\newblock \href {https://doi.org/10.1145/2903137} {\path{doi:10.1145/2903137}}.

\bibitem{dahal23half-integral}
Sameep Dahal and Jukka Suomela.
\newblock Distributed half-integral matching and beyond.
\newblock In Sergio Rajsbaum, Alkida Balliu, Joshua~J. Daymude, and Dennis
  Olivetti, editors, {\em Structural Information and Communication Complexity -
  30th International Colloquium, {SIROCCO} 2023, Alcal{\'{a}} de Henares,
  Spain, June 6-9, 2023, Proceedings}, volume 13892 of {\em Lecture Notes in
  Computer Science}, pages 339--356. Springer, 2023.
\newblock \href {https://doi.org/10.1007/978-3-031-32733-9_15}
  {\path{doi:10.1007/978-3-031-32733-9_15}}.

\bibitem{Improved-Deterministic-Distributed-Matching-via-Rounding}
Manuela Fischer.
\newblock Improved deterministic distributed matching via rounding.
\newblock {\em Distributed Comput.}, 33(3-4):279--291, 2020.
\newblock \href {https://doi.org/10.1007/s00446-018-0344-4}
  {\path{doi:10.1007/s00446-018-0344-4}}.

\bibitem{goos13local-approximation}
Mika G{\"{o}}{\"{o}}s, Juho Hirvonen, and Jukka Suomela.
\newblock Lower bounds for local approximation.
\newblock {\em J. {ACM}}, 60(5):39:1--39:23, 2013.
\newblock \href {https://doi.org/10.1145/2528405} {\path{doi:10.1145/2528405}}.

\bibitem{Linear-in-Delta}
Mika G{\"{o}}{\"{o}}s, Juho Hirvonen, and Jukka Suomela.
\newblock Linear-in-{$\Delta$} lower bounds in the {LOCAL} model.
\newblock {\em Distributed Comput.}, 30(5):325--338, 2017.
\newblock \href {https://doi.org/10.1007/s00446-015-0245-8}
  {\path{doi:10.1007/s00446-015-0245-8}}.

\bibitem{LCL-Trees-Landscape}
Christoph Grunau, V{\'{a}}clav Rozhon, and Sebastian Brandt.
\newblock The landscape of distributed complexities on trees and beyond.
\newblock In Alessia Milani and Philipp Woelfel, editors, {\em {PODC} '22:
  {ACM} Symposium on Principles of Distributed Computing, Salerno, Italy, July
  25 - 29, 2022}, pages 37--47. {ACM}, 2022.
\newblock \href {https://doi.org/10.1145/3519270.3538452}
  {\path{doi:10.1145/3519270.3538452}}.

\bibitem{MM2}
Michal Hanckowiak, Michal Karonski, and Alessandro Panconesi.
\newblock On the distributed complexity of computing maximal matchings.
\newblock {\em {SIAM} J. Discret. Math.}, 15(1):41--57, 2001.
\newblock \href {https://doi.org/10.1137/S0895480100373121}
  {\path{doi:10.1137/S0895480100373121}}.

\bibitem{Greedy-Optimal-Distributed-Maximal-Matching}
Juho Hirvonen and Jukka Suomela.
\newblock Distributed maximal matching: greedy is optimal.
\newblock In Darek Kowalski and Alessandro Panconesi, editors, {\em {ACM}
  Symposium on Principles of Distributed Computing, {PODC} '12, Funchal,
  Madeira, Portugal, July 16-18, 2012}, pages 165--174. {ACM}, 2012.
\newblock \href {https://doi.org/10.1145/2332432.2332464}
  {\path{doi:10.1145/2332432.2332464}}.

\bibitem{ISRAELI-ITAI}
Amos Israeli and Alon Itai.
\newblock A fast and simple randomized parallel algorithm for maximal matching.
\newblock {\em Inf. Process. Lett.}, 22(2):77--80, 1986.
\newblock \href {https://doi.org/10.1016/0020-0190(86)90144-4}
  {\path{doi:10.1016/0020-0190(86)90144-4}}.

\bibitem{Linial-Seminal-Paper}
Nathan Linial.
\newblock Locality in distributed graph algorithms.
\newblock {\em {SIAM} J. Comput.}, 21(1):193--201, 1992.
\newblock \href {https://doi.org/10.1137/0221015} {\path{doi:10.1137/0221015}}.

\bibitem{mayer95local}
Alain~J. Mayer, Moni Naor, and Larry~J. Stockmeyer.
\newblock Local computations on static and dynamic graphs (preliminary
  version).
\newblock In {\em Third Israel Symposium on Theory of Computing and Systems,
  {ISTCS} 1995, Tel Aviv, Israel, January 4-6, 1995, Proceedings}, pages
  268--278. {IEEE} Computer Society, 1995.
\newblock \href {https://doi.org/10.1109/ISTCS.1995.377023}
  {\path{doi:10.1109/ISTCS.1995.377023}}.

\bibitem{Naor1991}
Moni Naor.
\newblock A lower bound on probabilistic algorithms for distributive ring
  coloring.
\newblock {\em {SIAM} J. Discret. Math.}, 4(3):409--412, 1991.
\newblock \href {https://doi.org/10.1137/0404036} {\path{doi:10.1137/0404036}}.

\bibitem{naor-stockmeyer}
Moni Naor and Larry~J. Stockmeyer.
\newblock What can be computed locally?
\newblock {\em {SIAM} J. Comput.}, 24(6):1259--1277, 1995.
\newblock \href {https://doi.org/10.1137/S0097539793254571}
  {\path{doi:10.1137/S0097539793254571}}.

\bibitem{MM3}
Alessandro Panconesi and Romeo Rizzi.
\newblock Some simple distributed algorithms for sparse networks.
\newblock {\em Distributed Comput.}, 14(2):97--100, 2001.
\newblock \href {https://doi.org/10.1007/PL00008932}
  {\path{doi:10.1007/PL00008932}}.

\bibitem{Peleg-book}
David Peleg.
\newblock {\em Distributed Computing: A Locality-Sensitive Approach}.
\newblock Society for Industrial and Applied Mathematics, 2000.
\newblock \href {https://doi.org/10.1137/1.9780898719772}
  {\path{doi:10.1137/1.9780898719772}}.

\bibitem{schrijver-book}
A.~Schrijver.
\newblock {\em Combinatorial Optimization - Polyhedra and Efficiency}.
\newblock Springer, 2003.

\bibitem{yamashita96computing}
Masafumi Yamashita and Tsunehiko Kameda.
\newblock {Computing on anonymous networks: part {I}---characterizing the
  solvable cases}.
\newblock {\em {IEEE} Trans. Parallel Distributed Syst.}, 7(1):69--89, 1996.
\newblock \href {https://doi.org/10.1109/71.481599}
  {\path{doi:10.1109/71.481599}}.

\end{thebibliography}

\end{document}